\documentclass{article}
\usepackage{fullpage}

\usepackage[utf8]{inputenc}
\usepackage{amsthm,amssymb}
\usepackage{amsmath}
\usepackage{thmtools,mathtools}
\usepackage{caption,subcaption}
\usepackage{epsfig,graphicx,graphics,color,tikz,tikz-cd}
\usepackage [linesnumbered, ruled,vlined]{algorithm2e}
\usetikzlibrary{calc,decorations.pathmorphing,shapes}
\usepackage[most]{tcolorbox}
\usepackage{array,makecell,multirow}
\setcellgapes{4pt}
\usepackage{bbm}
\usepackage{varwidth}
\usepackage{rotating}
\usepackage{booktabs}
\usepackage[mathcal,mathscr]{eucal}
\usepackage{varwidth}
\usepackage{enumerate,enumitem}
\usepackage[nocompress,sort]{cite}
\usepackage{xspace}
\usepackage{bm}
\usepackage{csquotes}
\usepackage{hyperref}
\hypersetup{
colorlinks=true,       
linkcolor=blue,        
citecolor=magenta,         
filecolor=magenta,     
urlcolor=cyan,         
linktoc=all
}

\theoremstyle{plain}
\newtheorem{thm}{Theorem}[section]

\newtheorem{cor}[thm]{Corollary}
\newtheorem{cl}[thm]{Claim}

\theoremstyle{definition}

\newtheorem{rem}[thm]{Remark}

\def\final{0}  
\def\iflong{\iffalse}
\ifnum\final=0  
\newcommand{\knote}[1]{{\color{red}[{\tiny \textbf{Kristóf:} \bf #1}]\marginpar{\color{red}*}}}
\newcommand{\hnote}[1]{{\color{blue}[{\tiny \textbf{Hanna:} \bf #1}]\marginpar{\color{blue}*}}}
\newcommand{\ynote}[1]{{\color{purple}[{\tiny \textbf{Yuhang:} \bf #1}]\marginpar{\color{purple}*}}}
\else 
\newcommand{\knote}[1]{}
\newcommand{\hnote}[1]{}
\newcommand{\ynote}[1]{}
\fi

\makeatletter
\renewcommand{\paragraph}{%
  \@startsection{paragraph}{4}%
  {\z@}{1.5ex \@plus 1ex \@minus .2ex}{-1em}%
  {\normalfont\normalsize\bfseries}%
}
\makeatother

\newcommand{\bR}{\mathbb{R}}

\newcommand{\bZ}{\mathbb{Z}}

\newcommand{\cE}{\mathcal{E}}

\newcommand{\cH}{\mathcal{H}}
\newcommand{\cI}{\mathcal{I}}

\newcommand{\cM}{{\bm{M}}}
\newcommand{\cS}{\mathcal{S}}

\DeclareMathOperator{\supp}{supp}

\newcommand{\maxpf}{\textsc{Max-PF}\xspace}
\newcommand{\maxwpf}{\textsc{Max-WPF}\xspace}
\newcommand{\maxgpf}{\textsc{Max-GPF}\xspace}
\newcommand{\maxwgpf}{\textsc{Max-WGPF}\xspace}

\newcommand{\maxdbmis}{\textsc{Max-DBMIS}\xspace}
\newcommand{\maxwdbmis}{\textsc{Max-WDBMIS}\xspace}

\definecolor[named]{Blue}{cmyk}{1,0.1,0,0.1}
\definecolor[named]{Yellow}{cmyk}{0,0.16,1,0}
\definecolor[named]{Orange}{cmyk}{0,0.42,1,0.01}
\definecolor[named]{Red}{cmyk}{0,0.90,0.86,0}
\definecolor[named]{LightBlue}{cmyk}{0.49,0.01,0,0}
\definecolor[named]{Green}{cmyk}{0.20,0,1,0.19}

\SetKwFor{Whileblock}{while}{do}{}
\SetKwBlock{Begin}{}{}
\SetVlineSkip{1.5pt}
\makeatletter
\newcommand{\linkdest}[1]{\Hy@raisedlink{\hypertarget{#1}{}}}
\makeatother

\newlength{\bibitemsep}
\newlength{\bibparskip}
\let\oldthebibliography\thebibliography
\renewcommand\thebibliography[1]{%
\oldthebibliography{#1}%
\setlength{\parskip}{\bibitemsep}%
\setlength{\itemsep}{\bibparskip}%
}

\usepackage{titling}
\thanksmarkseries{alph}

\title{Approximating Maximum Properly Colored Forests\\ via Degree Bounded Independent Sets}

\author{
Yuhang Bai\thanks{School of Mathematics and Statistics, Northwestern Polytechnical University and Xi'an-Budapest Joint Research Center for Combinatorics, Xi'an 710129,
Shaanxi, People's Republic of China. Email: \texttt{yhbai@mail.nwpu.edu.cn}.}
\and
Kristóf Bérczi\thanks{MTA-ELTE Matroid Optimization Research Group and HUN-REN–ELTE Egerváry Research Group, Department of Operations Research, ELTE Eötvös Loránd University, and HUN-REN Alfréd Rényi Institute of Mathematics, Budapest, Hungary. Email: \texttt{kristof.berczi@ttk.elte.hu}.}
\and
Johanna K. Siemelink\thanks{Department of Operations Research, ELTE Eötvös Loránd University, Budapest, Hungary. Email: \texttt{hanna.siemelink@gmail.com}.}
}
\date{}

\begin{document}

\maketitle

\begin{abstract}

In the \emph{Maximum-size Properly Colored Forest} problem, we are given an edge-colored undirected graph and the goal is to find a properly colored forest with as many edges as possible. We study this problem within a broader framework by introducing the \emph{Maximum-size Degree Bounded Matroid Independent Set} problem: given a matroid, a hypergraph on its ground set with maximum degree $\Delta$, and an upper bound $g(e)$ for each hyperedge $e$, the task is to find a maximum-size independent set that contains at most $g(e)$ elements from each hyperedge $e$. We present approximation algorithms for this problem whose guarantees depend only on $\Delta$. When applied to the Maximum-size Properly Colored Forest problem, this yields a $2/3$-approximation on multigraphs, improving the $5/9$ factor of Bai, Bérczi, Csáji, and Schwarcz [\textit{Eur. J. Comb. 132 (2026) 104269}].

\medskip

\end{abstract}

\section{Introduction}
\label{sec:intro}

Throughout the paper, we consider loopless graphs that might contain parallel edges. A \emph{$k$-edge-colored} graph is a graph $G=(V,E)$ together with an edge-coloring $c\colon E\to [k]$. We call $G$ \emph{edge-colored} if it is $k$-edge-colored for some $k\in\mathbb{Z}_{\ge 1}$. A subgraph of $G$ is \emph{rainbow} if all its edges receive distinct colors, and \emph{properly colored} if any two of its adjacent edges have distinct colors.

Rainbow forests are exactly the common independent sets of the partition matroid defined by the color classes and the graphic matroid of the graph. Therefore a maximum-size rainbow forest can be found in polynomial time via Edmonds' matroid intersection algorithm~\cite{edmonds1970submodular}. By contrast, the properly colored case does not admit such a direct formulation, and is considerably less understood. Borozan, de La Vega, Manoussakis, Martinhon, Muthu, Pham, and Saad~\cite{borozan2019maximum} initiated the study of properly colored spanning trees in edge-colored graphs and investigated existence conditions. This problem generalizes the well-known \textit{Bounded-degree Spanning Tree} problem with uniform bounds, since the number of colors bounds the degree of each vertex. It also generalizes the properly colored \textit{Hamiltonian Path} problem when the number of colors is restricted to two. As both of these problems are NP-complete, finding a properly colored spanning tree is hard.

Bai, Bérczi, Csáji, and Schwarcz~\cite{bai2026approximating} introduced the \emph{Maximum-size Properly Colored Forest} problem (\maxpf), which asks for a properly colored forest of maximum size in an edge-colored graph, and analyzed its approximability on several graph classes, giving polynomial-time approximation algorithms and inapproximability bounds; we refer to the weighted variant of the problem as \emph{Maximum-weight Properly Colored Forest} problem (\maxwpf). In this paper, we extend their work by studying a more general problem formulation, which in turn yields improved approximation guarantees.

Let $\cM=(V,\mathcal{I})$ be a matroid on ground set $V$. Let $\mathcal{H}$ be a hypergraph on the same ground set with maximum degree $\Delta$, i.e., each element of $V$ belongs to at most $\Delta$ hyperedges, and let $g\colon E(\cH)\to\mathbb{Z}_{\ge 0}$ be upper bounds on the hyperedges. The \emph{Maximum-size Degree Bounded Matroid Independent Set} problem (\maxdbmis) asks to determine
\[
\max\{ |I| \colon I\in\mathcal{I}, |I\cap e|\le g(e) \text{ for all } e\in E(\cH) \}.
\]
We will also consider the weighted variant (\maxwdbmis), where each element $v\in V$ has a nonnegative weight $w(v)$, and the objective $|I|$ is replaced by $w(I)\coloneqq \sum_{v\in I} w(v)$. 
This problem is the independent set analogue of the \emph{Degree Bounded Matroid Basis} problem. For the latter, Király, Lau, and Singh~\cite{kiraly2012degree} showed that iterative relaxation yields a polynomial-time algorithm with additive violation at most $2\Delta-1$ for two-sided bounds, improving to $\Delta-1$ when only upper (respectively, only lower) bounds are present. 

\subsection{Our results}
\label{sec:our}

Bai, Bérczi, Csáji, and Schwarcz~\cite{bai2026approximating} showed that if a maximum-size set $U$ of edges can be covered by a union of matchings, where the $i$th matching consists of edges of color $i$, then there exists a maximum-size properly colored forest covering $U$. Any properly colored forest spanning $U$ has at least $|U|/2$ edges, giving a $1/2$ baseline. By performing local exchanges that replace many size-two components with larger properly colored components, they obtain a $5/9$-approximation algorithm for general graphs.

The above idea relies on the fact that a subgraph of an edge-colored graph is properly colored if and only if, for every color, the edges of that color form a matching. We now generalize this by allowing bounded degrees for every color. Given degree bounds $g_i\colon V\to\bZ_{\ge 0}$ for each $i\in[k]$, an edge set $F\subseteq E$ is \emph{$g$-properly colored} if, for each color $i\in[k]$ and each vertex $v\in V$, at most $g_i(v)$ edges of color $i$ in $F$ are incident to $v$. We consider an extension of \maxpf: finding a maximum-size $g$-properly colored forest (\maxgpf) or its weighted variant (\maxwgpf).

Observe that \maxgpf can be reduced to \maxdbmis as follows. For a $k$-edge-colored graph $G=(V,E)$, let $\cM=(E,\cI)$ be the graphic matroid on $E$. For each color $i\in[k]$, let $E_i\subseteq E$ denote the set of edges of color $i$, and for a vertex $v\in V$ let $\delta_{E_i}(v)$ be the edges of $E_i$ incident to $v$. We construct a hypergraph whose hyperedges are the sets $\delta_{E_i}(v)$ for all $v\in V$ and $i\in[k]$, each equipped with upper bound $g_i(v)$. A set $F\subseteq E$ is a $g$-properly colored forest in $G$ if and only if $F\in\cI$ and $|F\cap \delta_{E_i}(v)| \le g_i(v)$ holds for every $v\in V$ and every $i\in[k]$. In this reduction each edge belongs to at most two hyperedges, that is, $\Delta=2$.

We present approximation algorithms for the Degree Bounded Matroid Independent Set problem with guarantees depending on the maximum degree $\Delta$. For any fixed $\varepsilon>0$, we give polynomial-time algorithms achieving a $(2/(\Delta+1)-\varepsilon)$-approximation in the unweighted case (Theorem~\ref{thm:deltaMinusOne}) and a $1/(\Delta+\varepsilon)$-approximation in the weighted case with unit upper bounds, that is, $g(e)\le1$ for all $e\in E(\cH)$ (Theorem~\ref{thm:gmax1-delta}). As a corollary, we obtain a $2/3$-approximation for \maxgpf (Corollary~\ref{cor:maxpf-main}), improving the $5/9$ guarantee for \maxpf due to Bai, Bérczi, Csáji, and Schwarcz~\cite{bai2026approximating}. In the weighted case with unit upper bounds for \maxgpf, our guarantee becomes $1/(2+\varepsilon)$.
 
We also consider other properly colored structures. For the \textit{Maximum-size Properly Colored Forest with Bundles} problem, we devise a local-search algorithm that achieves a $1/3$-approximation (Theorem~\ref{thm:main}). In addition, we present a simple algorithm that attains a $3/4$-approximation for 2-edge-colored graphs and a $1/2$-approximation for 3-edge-colored graphs (Theorem~\ref{thm:SmallColors}). For directed graphs, we study $g$-properly colored and out-colored branchings: using the iterative refinement framework for three-matroid intersection~\cite{linhares2020iterative}, we obtain an LP-relative $1/2$-approximation for the \textit{Maximum-weight $g$-Properly Out-colored Branching} problem (Theorem~\ref{prop:outcoloredbranchings}), and, by viewing branchings as a special case of \maxwdbmis with $\Delta=3$, we derive a $1/2$-approximation in the unweighted case and a $(\ln 4)/5$-approximation in the weighted case (Theorem~\ref{thm:coloredbranchings}) for the \textit{Maximum-weight $g$-Properly Colored Branching} problem. Finally, for \textit{Maximum-weight $g$-Properly Colored $b$-matchings}, we show that an optimal solution can be found in strongly polynomial time via a reduction to the Maximum-weight Hierarchical $b$-matching problem studied by Madarasi~\cite{madarasi2021} (Theorem~\ref{thm:b-matchings}).

\subsection{Related work}
\label{sec:related}
Observe that the feasible sets for \maxdbmis form a $(\Delta+1)$-extendible system. Recall that an independence system $(V,\cI)$ is \emph{$k$-extendible} if, whenever $A\subseteq B\in\cI$ and $x\notin B$ with $A+x\in\cI$, there exists a set $Z\subseteq B\setminus A$ of size at most $k$ such that $(B+x)\setminus Z\in\cI$; informally, adding one element may require removing at most $k$ others. In our setting, $\cI$ consists of the sets that are independent in the matroid $\cM$ and respect all hyperedge upper bounds. When we try to add an element $x$, at most one element must be removed from each tight hyperedge containing $x$, and since every element lies in at most $\Delta$ hyperedges, this accounts for at most $\Delta$ deletions. If the resulting set violates matroid independence, we may need to remove one additional element to restore it. Thus at most $\Delta+1$ deletions are required, proving $(\Delta+1)$-extendibility. For $k$-extendible systems, the best known approximation guarantee is $1/k$~\cite{jenkyns1976efficacy}.

The notion of a $k$-exchange system, introduced by Feldman, Naor, Schwartz, and Ward~\cite{feldman2011improved}, strictly strengthens $k$-extendibility and admits better approximation guarantees through local search on squared weights (a technique originating with Berman~\cite{berman2000d}). This yields a $2/(k+1+\varepsilon)$ approximation for linear objectives and $2/(k+\varepsilon)$ for cardinality~\cite{feldman2011improved,berman2000d}. In addition, local search algorithms for $k$-set packing -- particularly those employing squared-weight improvements~\cite{berman2000d,berman2003optimizing,cygan2013improved,Neuwohner2021an,neuwohner2024limits,neuwohner2023passing,thiery2023improved} -- achieve a $3/(k+1)$ approximation in the unweighted case~\cite{cygan2013improved}. However, these analyses rely on structural properties such as strong base orderability, which do not hold in our setting. We emphasize that the feasible solutions of \maxdbmis do not form a $(\Delta+1)$-exchange system. Nonetheless, its packing constraints are uniform matroids (coming from the hyperedges), while only the underlying matroid is unrestricted, so the problem lies between $(\Delta+1)$-exchange systems and $(\Delta+1)$-extendible systems.

Another important problem that lies between $k$-exchange systems and $k$-extendible systems is the $k$-matroid intersection problem, a central topic in combinatorial optimization. When $k=1$, the problem reduces to finding an independent set in a single matroid and can be solved in polynomial time by the greedy algorithm. For $k=2$, Edmonds' celebrated algorithm~\cite{edmonds1970submodular} computes a maximum-size common independent set in polynomial time. For $k\ge 3$, however, no polynomial-time exact algorithm is known. The best known approximation ratio in the unweighted case is $2/k$, due to Lee, Sviridenko, and Vondrák~\cite{lee2010matroid}. Recently, Singer and Thiery~\cite{singer2025better} obtained a $\ln 4/(k+1)$-approximation for the weighted case, improving the longstanding $(k-1)$-approximation of Lee, Sviridenko, and Vondrák~\cite{lee2010submodular}.

\subsection{Notation and basic definitions}
\label{sec:notation}

\paragraph{Basic notation.}

For an integer $\ell$, we denote the \emph{set of integers greater or equal to $\ell$} by $\bZ_{\geq \ell}$. For a positive integer $k$, we use $[k]\coloneqq\{1,\dots,k\}$. Given a ground set $E$, the \emph{difference} of $X,Y\subseteq E$ is denoted by $X\setminus Y$. If $Y$ consists of a single element $y$, then $X\setminus \{y\}$ and $X\cup \{y\}$ are abbreviated as $X-y$ and $X+y$, respectively.

\paragraph{Graphs and hypergraphs.}

We consider loopless undirected graphs that may contain parallel edges. Let $G=(V,E)$ be a graph, $F\subseteq E$ a subset of edges, and $X\subseteq V$ a subset of vertices. We denote by $\delta_F(X)$ the set of edges in $F$ \textit{having exactly one endpoint} in $X$, while $F[X]$ denotes the set of edges in $F$ \textit{induced} by $X$. Given a function $b\colon V\to\bZ_{\ge 0}$, a \emph{$b$-matching} is a subset of edges $M\subseteq E$ such that, for every vertex $v\in V$, the number of edges of $M$ incident to $v$ is at most $b(v)$, that is, $|\delta_M(v)|\le b(v)$. This extends the notion of a matching, since $b\equiv 1$ yields an ordinary matching.

If the graph is directed, then the sets of edges \emph{entering} and the \emph{leaving} $X$ in $F$ are denoted by $\delta^{in}_F(X)$ and $\delta^{out}_F(X)$, respectively. An \emph{arborescence} is a directed tree with a root $r\in V$ that has in-degree zero, while every other vertex has in-degree one. A \emph{branching} is a union of vertex-disjoint arborescences.

Let $c\colon E\to[k]$ be an edge-coloring of $G$ using $k$ colors. The function $c$ is extended to subsets of edges by letting, for $F\subseteq E$, $c(F)$ denote the set of colors appearing on the edges of $F$. For an edge-colored graph $G=(V,E)$, we use $E_i=\{e\in E \colon c(e)=i\}$ to denote the set of edges of color $i$. Without loss of generality, we assume throughout that each $E_i$ contains no parallel edges.

A \emph{hypergraph} is a pair $\mathcal H=(V,E)$ where $E\subseteq 2^V$. For a set $F\subseteq E$ of hyperedges and $X\subseteq V$, we denote by $\delta_F(X)$ the set of hyperedges in $F$ that intersect both $X$ and $V\setminus X$.

\paragraph{Matroids.}

For basic definitions on matroids and on matroid optimization, we refer the reader to~\cite{oxley2011matroid,frank2011connections}. A \emph{matroid} $\cM=(E,\cI)$ is defined by its \emph{ground set} $E$ and its \emph{family of independent sets} $\cI\subseteq 2^E$ that satisfies the \emph{independence axioms}: (I1) $\emptyset\in\cI$, (I2) $X\subseteq Y,\ Y\in\cI\Rightarrow X\in\cI$, and (I3) $X,Y\in\cI,\ |X|<|Y|\Rightarrow\exists e\in Y\setminus X\ s.t.\ X+e\in\cI$. Members of $\cI$ are called \emph{independent}, while sets not in $\cI$ are called \emph{dependent}.

For a set $E$, the matroid in which every subset of $E$ is independent is called a \emph{free matroid}. For disjoint sets $E_1$ and $E_2$, the \emph{direct sum} $\cM_1\oplus \cM_2$ of matroids $\cM_1=(E_1,\cI_1)$ and $\cM_2=(E_2,\cI_2)$ is a matroid $\cM=(E_1\cup E_2,\mathcal{I})$ whose independent sets are the disjoint unions of an independent set of $\cM_1$ and an independent set of $\cM_2$. The intersection of $\cM_1$ and $\cM_2$ is denoted by $\cM_1\cap \cM_2$; note that this is not a matroid in general.

The {\it uniform matroid} $\cM=(E,\cI)$ of rank $r$ is defined as $\cI=\{I\subseteq E\colon |I|\leq r\}$. For a graph $G=(V,E)$, the \emph{graphic matroid} $M=(E,\cI)$ of $G$ is defined on the edge set by considering a subset $F\subseteq E$ to be independent if it is a forest, that is, $\cI=\{F\subseteq E\colon F\ \text{does not contain a cycle}\}$. 

\subsection{Paper organization}

The paper is organized as follows. Section~\ref{sec:dbmis} develops approximation algorithms for \maxwdbmis with guarantees depending on the maximum degree $\Delta$ of the hypergraph $\cH$.\footnote{In matroid algorithms, it is usually assumed that the matroid is given by an \emph{independence oracle} and the running time is measured by the number of oracle calls and other conventional elementary steps.} The remaining sections focus on approximation algorithms for other properly colored structures: in Section~\ref{sec:maxpf+}, we augment properly colored forests with bundles, Section~\ref{sec:arborescences} investigates $g$-properly colored arborescences in directed graphs, and Section~\ref{sec:bmatchings} presents algorithms for $g$-properly colored $b$-matchings. We conclude the paper by a list of open question in Section~\ref{sec:conc}.

\section{Degree bounded matroid independent set}\label{sec:dbmis}

This section is dedicated to proving our main results: Theorem~\ref{thm:deltaMinusOne} for the unweighted case and Theorem~\ref{thm:gmax1-delta} for the weighted case with unit upper bounds. Our proof proceeds via a reduction to the \emph{Matroid $k$-parity} problem, defined as follows. We are given a matroid $\cM=(V,\cI)$ and a collection $\cE$ of pairwise disjoint $k$-element subsets of $V$, together with a linear weight function $w\colon \cE\to\mathbb R_{\ge 0}$. The task is to find a maximum-weight subcollection $\cS\subseteq \cE$ such that $\bigcup_{e\in\cS} e\in\cI$. In the unweighted case we simply set $w(e)=1$ for all $e\in\cE$, so the objective becomes maximizing $|\cS|$.

\begin{thm}\label{thm:deltaMinusOne}
For any fixed $\varepsilon>0$, there exists a polynomial-time $(\tfrac{2}{\Delta+1}-\varepsilon)$-approximation algorithm for \maxdbmis, and a polynomial-time $\tfrac{\ln 4}{\Delta+2}$-approximation algorithm for \maxwdbmis, where $\Delta$ denotes the maximum degree of the input hypergraph.
\end{thm}
\begin{proof}
We begin with an arbitrary instance $(\cM=(V,\cI),\cH,g,w)$ of \maxwdbmis. Our goal is to reduce this instance, in polynomial time, to an instance of Matroid $(\Delta+1)$-parity and then apply an existing approximation algorithm for that problem.

For each $v\in V$, create $|\delta_\cH(v)|+1$ labeled copies of $v$, namely $v^{M}$ and $v^{e}$ for each $e\in\delta_{\cH}(v)$. Additionally, introduce a set $D_v$ of dummy elements with $|D_v|=\Delta-|\delta_\cH(v)|$ such that $D_u\cap D_v=\emptyset$ for $u\ne v$. We denote $V_M \coloneqq  \{v^M \colon  v\in V\}$, $V_e \coloneqq  \{v^{e} \colon v\in e\}$ for each $e \in E(\cH)$, and set
$$
V' \coloneqq  V_M \cup \bigl(\bigcup_{e\in E(\cH)} V_e\bigr) \cup \bigl(\bigcup_{v\in V} D_v\bigr).
$$
Now we define the collection $\cE' \coloneqq  \{e'_v \colon v\in V\}$ of disjoint $(\Delta+1)$-elements subsets of $V'$, where $e'_v \coloneqq  \{v^{M}\} \cup \{v^{e} \colon e\in\delta_{\cH}(v)\} \cup D_v$, so that $|e'_v|=\Delta+1$ for every $v\in V$. We also assign a weight $w'(e'_v)=w(v)$ for each $v\in V$.

Next we define matroids on $V'$ as follows. For
$
V_M \coloneqq  \{v^{M} \colon v\in V\},
$
we copy the original matroid: let $\phi_0\colon V_M\to V$ be given by $\phi_0(v^{M})=v$, and define the matroid $\cM_0$ on $V_M$ by $I\in\cI(\cM_0)$ if and only if $\phi_0(I)\in\cI(\cM)$,  so that $\cM_0\simeq\cM$. For each $e\in E(\cH)$, define a uniform matroid $\cM_e$ on $V_e$ having rank $g(e)$, thus enforcing the bounds. For each $v\in V$, let $\cM^{\mathrm{free}}_{v}$ be the free matroid on $D_v$. 
Finally, we take the direct sum 
$$
\cM_{\Sigma}
\coloneqq  \cM_0 \oplus \bigl(\bigoplus_{e\in E(\cH)}\cM_e\bigr) \oplus \bigl(\bigoplus_{v\in V}\cM^{\mathrm{free}}_{v}\bigr)
$$ on ground set $V'$. This way, we get a matroid $(\Delta+1)$-parity instance $(\cM_{\Sigma}=(V',\cI'),\cE',w')$.

Take first an arbitrary feasible solution $I\subseteq V$ for the original \maxwdbmis instance, that is, $I\in \cI(\cM)$ and $|I\cap e|\le g(e)$ for each $e\in E(\cH)$. Define $\mathcal S(I)\coloneqq \{e'_v\in\cE' \colon v\in I\}$. Then, $\{v^{M} \colon v\in I\}\in \cI(\cM_0)$, $\{v^{e} \colon v\in I\cap e\}\in \cI(\cM_e)$ for each $e\in E(\cH)$, and $\cM^{\mathrm{free}}_{v}$ is free for each $v\in V$. Consequently, $\bigl(\bigcup_{e'\in\mathcal S(I)}e'\bigr)\in\cI(\cM_{\sum})$. Moreover, by construction, we have $w'(\mathcal S(I)) = w(I)$.

Now let $\mathcal{S}\subseteq \cE'$ be a feasible solution of the $(\Delta+1)$-parity instance defined above. Define $I(\mathcal{S})\coloneqq\{v\colon e'_v\in\mathcal{S}\}$. Then, it is not difficult to check that $I(\mathcal{S})$ is a feasible solution for the original \maxwdbmis instance. Moreover, by construction, we have $w'(\mathcal{S})=w(I(\mathcal{S}))$.

By the above reduction to matroid $(\Delta+1)$-parity, known approximation guarantees for matroid $k$-parity with $k=\Delta+1$ immediately yield a $(2/(\Delta+1) - \varepsilon)$-approximation for \maxdbmis~\cite{lee2010matroid} and a $(\ln 4)/(\Delta+2)$-approximation for \maxwdbmis~\cite{singer2025better}. The reduction takes polynomial time and relies only on independence oracles for $\cM$; the uniform and free matroids in the construction have trivial independence oracles. 
\end{proof}

By the reduction given in Section~\ref{sec:our}, we obtain the following corollary.

\begin{cor}\label{cor:maxpf-main}
    For any fixed $\varepsilon>0$, there is a polynomial-time $(2/3-\varepsilon)$-approximation algorithm for \maxgpf.
\end{cor}

\begin{rem}
The natural LP relaxation for \maxpf has integrality gap $2/3$ already on the simplest instance: a monochromatic triangle with unit weights. The integral optimum is $1$, since at most one edge can be chosen, whereas the fractional solution $x_e=1/2$ for all three edges is feasible and has value $2/3$.
\end{rem}

When every hyperedge has upper bound $g(e)\le 1$, \maxwdbmis admits an approximation guarantee arbitrarily close to $1/\Delta$.

\begin{thm}\label{thm:gmax1-delta}
For any fixed $\varepsilon>0$, there exists a polynomial-time $1/(\Delta+\varepsilon)$-approximation algorithm for \maxwdbmis when $g(e)\le 1$ for all $e\in E(\cH)$, where $\Delta$ denotes the maximum degree of the input hypergraph.
\end{thm}
\begin{proof}
Consider an instance $(\cM=(V,\cI),\cH,g,w)$ of \maxwdbmis. Let $S\subseteq V$ be a feasible solution, and let $S^\star$ be an optimal solution. 

We construct a bipartite \emph{conflict graph} $G=(A,B)$ as follows. Let $A=\{\alpha_v\colon v\in S\}$ and $B=\{\beta_u\colon u\in S^\star\}$. We add an edge $\{\alpha_v,\beta_u\}$ to $G$ if and only if there exists $e\in E(\cH)$ with $\{u,v\}\subseteq e$. Since $g(e)\le 1$, for each fixed $v\in S$ and each hyperedge $e$ containing $v$, there is at most one $u\in S^\star\cap e$. As each element lies in at most $\Delta$ hyperedges, every vertex of $G$ has degree at most $\Delta$. By Kőnig's edge coloring theorem~\cite{konig1916graphen} for bipartite graphs, we can partition $E(G)$ into $\Delta$ disjoint matchings,
\[
E(G)=N_1\ \dot\cup\ \cdots\ \dot\cup\ N_\Delta.
\]
Define $\Delta$ partition matroids and one restriction of $\cM$ on the common ground set $S\cup S^\star$ as follows. For each $i\in[\Delta]$, let $\cM_i$ be the partition matroid whose parts are the pairs $\{v,u\}$ corresponding to the conflict edges $\alpha_v\beta_u\in N_i$, each part having rank $1$. Thus any $\cM_i$-independent set contains at most one endpoint of every conflict edge in $N_i$. Let $\cM_0\coloneqq \cM|_{S\cup S^\star}$ be the restriction of $\cM$ to $S\cup S^\star$.

By construction, a set $X\subseteq S\cup S^\star$ is independent in the intersection $\cM_0 \cap \bigl(\bigcap_{i=1}^{\Delta}\cM_i\bigr)$ if and only if $X\in\cI$ and for every matching $N_i$, the set $X$ avoids selecting both endpoints of any conflict edge in $N_i$. Consequently, the feasible subsets of $S\cup S^\star$ for \maxwdbmis are precisely the independent sets of $\cM_0 \cap \bigl(\bigcap_{i=1}^{\Delta}\cM_i\bigr)$.

Let $k\coloneqq \Delta+1$ and $p\coloneqq \lceil 1/\varepsilon\rceil$. The $p$-exchange local search analysis of~\cite{lee2010submodular} for maximizing a linear objective function over the intersection of $k$ matroids implies that any $p$-local optimum $X$ for $\cM_0\cap\bigl(\bigcap_{i=1}^{\Delta}\cM_i\bigr)$ satisfies
\[
w(X)\ \ge\ \frac{1}{\,k-1+1/p\,}\,w(X^\star)\ =\ \frac{1}{\,\Delta+1/p\,}\,w(X^\star)\ =\ \frac{1}{\,\Delta+\varepsilon\,}\,w(X^\star),
\]
where $X^\star$ is an optimal solution in the intersection. Applying this exchange argument to $S\cup S^\star$ implies that, unless $w(S)\ge \frac{1}{\Delta+\varepsilon}w(S^\star)$, there exists a $p$-exchange for $S$ in the underlying \maxwdbmis instance. Therefore the standard $p$-exchange local search on \maxwdbmis terminates in polynomial time for fixed $p$ and returns a solution $S$ with $w(S)\ \ge\ \frac{1}{\Delta+\varepsilon}\,w(S^\star)$.
\end{proof}

\begin{rem}
For unit upper bounds, we obtained a $1/(\Delta+\varepsilon)$-approximation for $\maxwdbmis$ under an arbitrary matroid constraint. In the proof of Theorem~\ref{thm:deltaMinusOne}, we reduced $\maxwdbmis$ to weighted matroid $k$-parity. For the instances produced by this reduction, we have $k\coloneqq \Delta+1$, and thus the above guarantee yields a $1/(k-1+\varepsilon)$-approximation for this special class of weighted matroid $k$-parity instances. Compared with the general $\ln 4/(k+1)$ bound, observe that $1/(k-1+\varepsilon)>\ln 4/(k+1)$ whenever $k\le6$ (equivalently, $\Delta\le5$). Hence, for small $k$, our unit upper  bound strictly improves the general guarantee on this subclass. Moreover, to the best of our knowledge, even for general weighted matroid $k$-parity no $1/(k-1)$-approximation is currently known; our result therefore achieves this benchmark on the above subclass when $k\le6$.
\end{rem}

\begin{rem}
The $O(1/\Delta)$ guarantees are tight up to constant factors. The well-known \emph{$k$-Dimensional Matching} problem is a special case of \maxdbmis obtained by taking the free matroid, assigning upper bound $1$ to every vertex on its incident hyperedges, and noting that each chosen hyperedge participates in exactly $k$ such constraints, that is, $\Delta=k$. Lee, Svensson, and Thiery~\cite{lee2025asymptotically} proved that, unless $\mathrm{NP}\subseteq\mathrm{BPP}$, for any $\varepsilon>0$ and all sufficiently large $k$, no polynomial-time algorithm can approximate $k$-Dimensional Matching within a factor better than $(12+\varepsilon)/k$.
\end{rem}

\section{Other structures}
\label{sec:properly}

After establishing our results for forests, a natural next step is to consider other properly colored structures within the same framework. We examine three directions in this section. Section~\ref{sec:maxpf+} augments properly colored forests with bundles. Section~\ref{sec:arborescences} investigates $g$-properly colored arborescences in directed graphs, integrating color and branching constraints. Section~\ref{sec:bmatchings} develops $g$-properly colored $b$-matchings. 

\subsection{Properly colored forests with bundles}
\label{sec:maxpf+}

Let $G=(V,E)$ be a graph. For distinct $u, v\in V$, denote $E_{uv} \coloneqq  \{e\in E \colon e \text{ has endpoints $u$ and $v$}\}$. A \emph{bundle} between $u$ and $v$ is any nonempty subset of $E_{uv}$. Given a set $F\subseteq E$, we denote by $\supp(F)$ the \emph{support graph} of $F$, obtained by substituting each bundle of parallel edges by a single edge. Then, $F$ is a \emph{forest with bundles} if $\supp(F)$ is a forest, or equivalently, $F$ contains no cycle of length at least three. The notion is inspired by \maxwpf, where a weighted edge can be viewed as representing several parallel edges between its endpoints. The analogy is not literal, since it is not clear how to define the colors of these edges; still, this idea motivates the model. In the \emph{Maximum-size Properly Colored Forest with Bundles} problem, the goal is to find a maximum-size properly colored forest with bundles. The resulting problem, however, is genuinely different from \maxwpf, and we know of no reduction in either direction. In particular, forests with bundles do not form the independent sets of a matroid, hence the $1/2$-approximation barrier does not carry over. In fact, the family of feasible solutions is not even $3$-extendable. 

Nevertheless, we can still provide a $1/3$-approximation algorithm. Its high-level idea is as follows. The algorithm proceeds via local improvements. Starting from an arbitrary properly colored forest with bundles, it repeatedly looks for ways to increase its size without introducing a cycle in the support graph or violating the coloring constraints. It first performs greedy augmentations, adding any edge whose inclusion keeps the set feasible. Then, for each pair of vertices $u$ and $v$ connected by at least one unused edge, the algorithm considers a local exchange: if $u$ and $v$ are already connected by a path in the support graph, adding an edge between them would create a cycle. In this case, it inspects the edges of this path and, for each such edge, considers removing its entire bundle and replacing it with a larger bundle between $u$ and $v$ whose colors do not conflict with those of the remaining edges. Whenever such an improving exchange is possible, the solution is updated and the greedy augmentation step is resumed. The procedure terminates when no further improvement is possible, yielding a properly colored forest with bundles. The algorithm is presented as Algorithm~\ref{algo:maxpf+}.

\begin{algorithm}[h]
\caption{Local search algorithm for properly colored forests with bundles.}\label{algo:maxpf+}
\DontPrintSemicolon

\KwIn{A graph $G=(V,E)$ with edge-coloring $c\colon E\to[k]$.}
\KwOut{A properly colored forest with bundles $F$ in $G$.}

\medskip

Let $F$ be any properly colored forest with bundles.\label{step_+:1}\;
\For{$e \in E\setminus F$\label{step_+:for1}}{
\begin{varwidth}{0.9\linewidth}
If $F +e$ is a properly colored forest with bundles, then $F \gets F +e$.  
\end{varwidth}\; 
}
\For{each $u,v\in V$ with $E_{uv}\setminus F\neq\emptyset$\label{step_+:for2}}{
Let $E^\circ_{uv} \coloneqq E_{uv}\setminus F$.\;
\If{there exists a path $u=z_0,z_1,\dots,z_q=v$ in $\supp(F)$}{
    \For{each $i\in[q]$}{
        Let $S\coloneqq E_{z_{i-1}z_i}\cap F$.\;
        Let $E^\bullet_{uv}\coloneqq\bigl\{e\in E^\circ_{uv}\colon c(e)\notin c\bigl(\delta_{F\setminus S}(u)\cup \delta_{F\setminus S}(v)\bigr)\bigr\}$.\;
        \If{$|E^\bullet_{uv}|>|S|$}{
            $F\gets (F\setminus S)\cup E^\bullet_{uv}$ and go to Step~\ref{step_+:for1}.\;
        }
    }
}
}
\Return{$F$}\;
\end{algorithm} 

\begin{thm}\label{thm:main}
Algorithm~\ref{algo:maxpf+} provides a $1/3$-approximation for the Maximum-size Properly Colored Forest with Bundles problem in polynomial time.
\end{thm}
\begin{proof}
    Let $F$ be the solution returned by Algorithm~\ref{algo:maxpf+} and let $O$ be a maximum-size properly colored forest with bundles. By the construction of Algorithm~\ref{algo:maxpf+}, $F$ is a feasible solution. For each unordered pair $u,v\in V$, we write $F_{uv}\coloneqq F\cap E_{uv}$ and $O_{uv}\coloneqq O\cap E_{uv}$. First, we show that the solution returned by the algorithm is locally optimal with respect to replacing any subset of edges by a bundle between two vertices.
    
    \begin{cl}\label{cl:change}
        For each pair $u,v\in V$ and every $E^\bullet_{uv} \subseteq E^\circ_{uv}\subseteq E\setminus F$,
        \begin{equation}\label{eq:local-opt}
        \min\Bigl\{|S| \colon \ (F\setminus S)\cup E^\bullet_{uv} \text{ is a properly colored forest with bundles}\Bigr\}\ \ge\ |E^\bullet_{uv}|.
        \end{equation}
    \end{cl}
    \begin{proof}
        Fix a pair $u,v\in V$ and let $E^\bullet_{uv}\subseteq E^\circ_{uv}$. If $E^\bullet_{uv}=\emptyset$, then \eqref{eq:local-opt} clearly holds, so assume $E^\bullet_{uv}\neq\emptyset$. It suffices to show that for every set $S\subseteq F$ for which $(F\setminus S)\cup E^\bullet_{uv}$ is a properly colored forest with bundles, we have $|S|\ge |E^\bullet_{uv}|$, since taking the minimum over all such sets $S$ then yields \eqref{eq:local-opt}. We distinguish two cases according to the structure of $\supp(F)$.
        
        \medskip
        \noindent\textbf{Case~1.} There is no $u$-$v$ path in $\supp(F)$ of length at least two.
        \smallskip
        
        In this case, either $u$ and $v$ lie in different connected components of $\supp(F)$, or $uv$ is the unique edge on the $u$-$v$ path in $\supp(F)$. In either situation, adding any edge $e\in E_{uv}$ to $F$ cannot create a cycle in $\supp(F)$.
        
        By Step~\ref{step_+:for1} of Algorithm~\ref{algo:maxpf+}, for every edge $e\in E\setminus F$ the graph $F+e$ is not a properly colored forest with bundles. Hence, for each $e\in E^\bullet_{uv}$ there exists an edge $f_e\in F$ incident with $u$ or $v$ such that $c(f_e)=c(e)$. Since $(F\setminus S)\cup E^\bullet_{uv}$ is properly colored and every edge in $E^\bullet_{uv}$ is incident with both $u$ and $v$, the colors of the edges in $E^\bullet_{uv}$ are pairwise distinct. As $F$ is also properly colored, for each color there is at most one edge of that color incident with $u$ and at most one incident with $v$. It follows that the edges $f_e$ are all distinct, and the map $e \mapsto f_e$ is injective from $E^\bullet_{uv}$ into $F$. Moreover, if some $f_e\notin S$, then $f_e$ and $e$ would both be present in $(F\setminus S)\cup E^\bullet_{uv}$, sharing an endpoint and a color, contradicting proper coloring. Thus $f_e\in S$ for all $e\in E^\bullet_{uv}$, and hence
        \[
        |S| \ge |\{f_e \colon e\in E^\bullet_{uv}\}|  = |E^\bullet_{uv}|.
        \]
        
        \medskip
        \noindent\textbf{Case~2.} There is a $u$-$v$ path of length at least two in $\supp(F)$.
        \small
        
        Since $\supp(F)$ is a forest, there is a unique simple path between $u$ and $v$ in $\supp(F)$; write it as $u=z_0,z_1,\dots,z_q=v$, where $q\ge 2$. For each $i\in[q]$, denote by $S_i \coloneqq F_{z_{i-1}z_i}$ the bundle of edges in $F$ between $z_{i-1}$ and $z_i$.
        
        Because $(F\setminus S)\cup E^\bullet_{uv}$ is a forest with bundles and $E^\bullet_{uv}\neq\emptyset$, its support graph $\supp((F\setminus S)\cup E^\bullet_{uv})$ contains an edge between $u$ and $v$. If $S$ did not contain all edges of at least one bundle $S_i$, then for every $i$ there would remain at least one edge of $F$ between $z_{i-1}$ and $z_i$ after deleting $S$, and thus the entire path $u=z_0,z_1,\dots,z_q=v$ would still be present in $\supp(F\setminus S)$. Adding any edge between $u$ and $v$ would then create a cycle of length at least three in the support graph, contradicting that $(F\setminus S)\cup E^\bullet_{uv}$ is a forest with bundles. Hence there exists an index $i\in[q]$ such that $S_i \subseteq S$. Fix such an index $i$, and write $S' \coloneqq S\setminus S_i$, so that $|S| = |S_i| + |S'|$. For this fixed $i$, Algorithm~\ref{algo:maxpf+} considers the set
        \[
        E^{\bullet,i}_{uv} \coloneqq \bigl\{e\in E_{uv}\setminus F \colon c(e)\notin c(\delta_{F\setminus S_i}(\{u,v\}))\bigr\}.
        \]
        By Step~\ref{step_+:for2}, we have
        \begin{equation}\label{eq:algo-ineq}
        |E^{\bullet,i}_{uv}| \le |S_i|,
        \end{equation}
        for otherwise the algorithm would have performed the exchange
        \(
        F \gets (F\setminus S_i)\cup E^{\bullet,i}_{uv}
        \). Now define
        \[
        T \coloneqq E^\bullet_{uv} \cap E^{\bullet,i}_{uv}.
        \]
        Since $T\subseteq E^{\bullet,i}_{uv}$, we obtain from \eqref{eq:algo-ineq} that
        \begin{equation}\label{eq:T-bound}
        |T| \le |E^{\bullet,i}_{uv}| \le |S_i|.
        \end{equation}
        
        Consider any edge $e\in E^\bullet_{uv}\setminus T$. By definition of $T$ and $E^{\bullet,i}_{uv}$, we have $c(e)\in c\bigl(\delta_{F\setminus S_i}(\{u,v\})\bigr),$
        so there exists an edge $f_e\in \delta_{F\setminus S_i}(u)\cup \delta_{F\setminus S_i}(v)$
        with $c(f_e)=c(e)$. Since $f_e\notin S_i$ and $(F\setminus S)\cup E^\bullet_{uv}$ is properly colored, such an edge $f_e$ must be removed by $S'$, that is, $f_e\in S'$.
        
        As before,  the colors $c(e)$ for $e\in E^\bullet_{uv}\setminus T$ are all distinct. Since $F$ is properly colored, for each such color there is at most one edge of that color incident with $u$ or $v$ in $F\setminus S_i$. Hence the edges $f_e$ are all distinct and the map $e \longmapsto f_e$
        is injective from $E^\bullet_{uv}\setminus T$ into $S'$. Therefore,
        \begin{equation}\label{eq:Sprime-bound}
        |S'|\ \ge\ |E^\bullet_{uv}\setminus T|.
        \end{equation}
        Combining \eqref{eq:T-bound} and \eqref{eq:Sprime-bound}, we obtain
        \[
        |S| = |S_i| + |S'|
         \ge |T| + |E^\bullet_{uv}\setminus T|
         = |E^\bullet_{uv}|.
        \]
        
        \bigskip
        In both cases we have shown that every set $S\subseteq F$ for which $(F\setminus S)\cup E^\bullet_{uv}$ is a properly colored forest with bundles satisfies $|S|\ge |E^\bullet_{uv}|$, which implies \eqref{eq:local-opt}. This completes the proof.
    \end{proof}
    
    For $u,v\in V$, let $\Delta_{uv}\coloneqq O_{uv}\setminus F_{uv}$, that is, $\Delta_{uv}$ is the edges of $O$ from $E_{uv}$ that are not in $F$. By Claim~\ref{cl:change}, summing \eqref{eq:local-opt} over all pairs with $\Delta_{uv}\neq\emptyset$ yields
    \begin{equation}\label{eq:sum-local-opt}
    |O\setminus F|\ =\ \sum_{u,v\in V} |\Delta_{uv}| \le \sum_{u,v\in V} \sigma_{uv},
    \end{equation}
    where $\sigma_{uv}$ denotes the minimum size of set $S_{uv}\subseteq F$ such that  $(F\setminus S)\cup\Delta_{uv}$ is feasible. For each pair $u,v\in V$, let $S_{uv}$ be such a set. 
    We set 
    \[
    S^{\mathrm{col}}_{uv}\coloneqq\{f\in S\colon f\in\delta_F(u)\cup\delta_F(v)\ \text{and there exists $e\in\Delta_{uv}$ with}\ c(e)=c(f)\}
    \]
    and
    \[
    S^{\mathrm{cyc}}_{uv}\coloneqq\{f\in S\colon f\ \text{is in a bundle on the unique $u$-$v$ path in $\supp(F)$}\}.
    \]
    Note that the sets $S^{\mathrm{col}}_{uv}$ and $S^{\mathrm{cyc}}_{uv}$ are not necessarily disjoint, and by the minimality of $S_{uv}$ we have $S_{uv}=S^{\mathrm{col}}_{uv}\cup S^{\mathrm{cyc}}_{uv}$.

    For each edge $e\in S^{\mathrm{col}}_{uv}$, charge $e$ to the endpoint $x$ at which the color conflict occurs, namely the vertex $x$ such that $\Delta_{uv}$ contains an edge of color $\gamma$ incident to $x$ and $F$ contains an edge of the same color incident to $x$. For a fixed vertex $x$ and a fixed color $\gamma$, this can happen for at most one pair that charges $x$, because $O$ is properly colored. Thus, each vertex is charged at most once. Consequently, every edge in $F\setminus O$ is charged at most twice, whereas edges in $F\cap O$ are never charged. Therefore,
    \begin{equation}\label{eq:color-total}
    \sum_{u,v\in V} |S^{\mathrm{col}}_{uv}| \le 2|F\setminus O|.
    \end{equation}

    Adding $\Delta_{uv}$ can create a cycle only when $u$ and $v$ lie in the same connected component of $\supp(F)$. In that case, removing one edge on the unique $u$-$v$ path in $\supp(F)$ restores acyclicity for all edges in $\Delta_{uv}$. Fix a pair with $uv\in \supp(O)$, and let $T$ be the vertex set of the connected component of $\supp(F)$ containing $u$ and $v$. By the exchange property of matroids, the edges of $\supp(O[T])$ can be matched injectively to edges of $\supp(F[T])$ in such a way that each edge $xy$ in $\supp(O[T])$ is assigned an edge on the unique $x$-$y$ path in $\supp(F[T])$. We then include in $S^{\mathrm{cyc}}_{uv}$ all edges of $F\setminus O$ in the bundle corresponding to $\phi_T(uv)$ on the $u$-$v$ path in $\supp(F[T])$.
    Since $\phi_T$ is injective, we get
    \begin{equation}\label{eq:cycle-total}
    \sum_{u,v\in V} |S^{\mathrm{cyc}}_{uv}| \le |F\setminus O|.
    \end{equation}
    
    For every pair $u,v\in V$, the $(F\setminus S_{uv})\cup\Delta_{uv}$ is feasible. Hence $\sigma_{uv}\le |S_{uv}| \le |S^{\mathrm{col}}_{uv}|+|S^{\mathrm{cyc}}_{uv}|$. Summing over all pairs and using \eqref{eq:sum-local-opt}, \eqref{eq:color-total}, and \eqref{eq:cycle-total}, we obtain
    \[
    |O\setminus F| \le \sum_{u,v\in V} \sigma_{uv} \le \sum_{u,v\in V} |S^{\mathrm{col}}_{uv}| + \sum_{u,v\in V} |S^{\mathrm{cyc}}_{uv}| \le 2|F\setminus O| + |F\setminus O| = 3|F\setminus O|.
    \]
    This yields
    \[
    |O|=|O\cap F|+|O\setminus F| \le |O\cap F| + 3|F\setminus O| =\ 3|F| - 2|F\cap O| \le 3|F|,
    \]
    proving the approximation factor.

    Regarding time complexity, Steps~\ref{step_+:1}--\ref{step_+:for1} only require, for each edge $e\in E$, to test whether $F+e$ is a properly colored forest with bundles. This can clearly be done in polynomial time. The only nontrivial part is Step~\ref{step_+:for2}. Fix a pair $u,v\in V$ with $E_{uv}^\circ E_{uv}\setminus F\neq\emptyset$. Since $F$ is a forest with bundles, the support graph $\supp(F)$ is a forest, so there is a unique simple $u$-$v$ path in $\supp(F)$, if any. The algorithm either finds this path or determines that it does not exist, and in the former case it iterates over the bundles on this path. For each such bundle $S$, it constructs the corresponding candidate set $E_{uv}^\bullet$ by scanning the edges incident with $u$ or $v$ in $F\setminus S$ and the edges in $E_{uv}^\circ$, and checks whether $|E_{uv}^\bullet|>|S|$. All these operations are straightforward to implement in polynomial time. Therefore, for a fixed $F$, a full execution of Step~\ref{step_+:for2} also takes polynomial time. It remains to bound the number of times the algorithm restarts from Step~\ref{step_+:for1}. Whenever this happens, Step~\ref{step_+:for2} has found a local improvement with some $S$ and $E_{uv}^\bullet$ satisfying $|E_{uv}^\bullet|>|S|$, and we update $F \gets (F\setminus S)\cup E_{uv}^\bullet$. Hence the size $|F|$ strictly increases at each such restart. Since $F\subseteq E$, we always have $|F|\le |E|$, so the algorithm can return to Step~\ref{step_+:for1} at most $|E|$ times. Combining this with the fact that each iteration takes polynomial time, we conclude that Algorithm~\ref{algo:maxpf+} runs in polynomial time.
\end{proof}

When the number of colors is small, the approximation factor can be further improved.

\begin{thm}\label{thm:SmallColors}
    For the Maximum-size Properly Colored Forest with Bundles problem on $k$-edge-colored graphs,
    there exists a polynomial-time algorithm that achieves a $3/4$-approximation when $k=2$ and a $1/2$-approximation when $k=3$.
\end{thm}
\begin{proof}
    We use the following simple algorithm: for each color $i\in[k]$, compute a maximum matching $M_i\subseteq E_i$ and take their union $M\coloneqq \bigcup_{i=1}^k M_i$. Then  return a maximum forest with bundles in $M$. Adapting the analysis in~\cite{bai2026approximating}, this yields a $3/4$-approximation for 2-edge-colored graphs and a $1/2$-approximation for 3-edge-colored graphs.
\end{proof}

\subsection{Branchings}
\label{sec:arborescences}

It is natural to ask for a directed analogue of \maxpf. In the directed setting, we consider two types of color constraints, depending on whether the bounds apply to all incident arcs or only to outgoing arcs; note that, since the in-degree of every vertex is at most one in a branching, bounds on incoming arcs are not meaningful. In the \emph{$g$-properly colored} setting, for every color $i$ and vertex $v$, at most $g_i(v)$ selected arcs of color $i$ may be incident to $v$. In the \emph{$g$-properly out-colored} setting, the bounds apply only to outgoing arcs, so for all $v$ and $i$ the number of selected arcs of color $i$ leaving $v$ is at most $g_i(v)$.

Applying the iterative refinement algorithm of Linhares, Olver, Swamy, and Zenklusen~\cite{linhares2020iterative} for three-matroid intersection yields a $1/2$-approximation for the Maximum-weight $g$-Properly Out-colored Branchings problem.

\begin{thm}\label{prop:outcoloredbranchings}
There is a polynomial-time LP-relative $1/2$-approximation algorithm for the Maximum-weight $g$-Properly Out-colored Branchings problem.
\end{thm}
\begin{proof}
Let $G=(V,E)$ be a directed graph, $w\colon E\to\bR$ be a weight function, and $g\colon V\to\bZ_{\geq 0}$ be upper bounds. We reduce the problem to an instance of three-matroid intersection on the arc set $E$. Let $\cM_G$ be the graphic matroid of the underlying undirected graph, and let $\cM_{\mathrm{in}}$ be the partition matroid whose parts are the sets $\delta^{in}(v)$ for $v\in V$, each with upper bound $1$. The intersection $\cM_G \cap \cM_{\mathrm{in}}$ therefore describes all branchings. To enforce the color bounds, for each vertex $v$ introduce a partition matroid $\cM_{v}$ on $\delta^{out}_E(v)$ whose parts are the sets $\delta^{out}_{E_i}(v)$ and whose upper bounds are $g_i(v)$, where recall that $E_i$ is the set of arcs of color $i$. Since every arc has a unique tail, these matroids are defined on disjoint subsets of $E$ and together form the direct sum $\cM_g = \bigoplus_{v\in V} \cM_{v}$. The independent sets of the intersection $\cM_G \cap \cM_{\mathrm{in}} \cap \cM_g$ are thus exactly the $g$-properly out-colored branchings.

Using the algorithm of~\cite{linhares2020iterative} for the intersection of three matroids, applied to the matroids $\cM_G$, $\cM_{\mathrm{in}}$, and $\cM_g$, yields an LP-relative $1/2$-approximation for finding a maximum-weight $g$-properly out-colored branching.
\end{proof}

The Maximum-weight $g$-Properly Colored Branchings problem, however, reduces to \maxwdbmis.

\begin{thm}\label{thm:coloredbranchings}
    The Maximum-size $g$-Properly Colored Branchings problem admits a polynomial-time $1/2$-approximation. Moreover, the Maximum-weight $g$-Properly Colored Branchings problem admits a $(\ln 4)/5$-approximation.
\end{thm}
\begin{proof}
    It is not difficult to see that the Maximum-weight $g$-Properly Colored Branchings problem is a special case of \maxwdbmis with $\Delta = 3$. Then, by Theorem~\ref{thm:deltaMinusOne}, the statement follows.
\end{proof}

\begin{rem}
The \emph{Rainbow Arborescence Conjecture}, due to Yokoi~\cite[Open Problem: Rainbow Arborescence Problem]{de2019combinatorial}, asserts that any digraph on $n$ vertices that is the union of $n-1$ spanning arborescences contains an arborescence using exactly one arc from each. Although several partial results are known~\cite{berczi2024rainbow,berczi2025rainbow}, the conjecture remains open in full generality. As a relaxation, one may instead ask for the existence of a properly out-colored or properly colored arborescence in the same setting; we refer to the latter as the \emph{Properly Colored Arborescence Conjecture}.
\end{rem}

\subsection{\texorpdfstring{$b$}{b}-matchings}
\label{sec:bmatchings}

Straying further from forests we consider maximum weight $g$-properly colored $b$-matchings. This problem can be addressed using a strongly polynomial algorithm for \emph{simultaneous assignment on weighted hierarchical $b$-matchings} due to Madarasi~\cite[Corollary 3]{madarasi2021}. In that problem, one is given a graph $G' = (V', E')$ and weights $w'\colon V \rightarrow \bR_+ $ together with a laminar family $\mathcal L$ of vertex subsets. For each $L\in\mathcal L$, an upper bound $g'\colon \mathcal{L} \rightarrow \bZ_{\geq 0}$  is imposed on the total degree of the vertices in $L$. The objective is to find a set of edges of maximum total weight that satisfies all these degree constraints.

\begin{thm}\label{thm:b-matchings} 
A maximum-weight $g$-properly colored $b$-matching can be found in strongly polynomial time.
\end{thm}
\begin{proof}
   Let $(G = (V, E), w, g)$ be a maximum-weight $g$-properly colored $b$-matching instance. We construct simultaneous assignment on weighted hierarchical $b$-matchings instance $(G' = (V', E'),w', \mathcal{L}, g')$ as follows. 

    For each vertex $v\in V$ and for each edge $e\in\delta(v)$ incident to $v$, create a copy $v^e$ and include it in $V'$. For every edge $e=uv\in E$, add to $E'$ the edge $u^e v^e$, so $G'$ contains one edge corresponding to each original edge; we define the weight of this edge to be $w'(u^ev^e)=w(uv)$. For every vertex $v\in V$ and color $i\in[k]$, define the vertex set $L_{v,i} \coloneqq \{v^e \in V' \colon e \in \delta_{E_i}(v)\}$, where $E_i$ is the set of edges of color $i$ in $G$. The color constraint $|\delta_{F_i}(v)| \le g_i(v)$ for a selected edge set $F$ is equivalent to requiring that the degree contributed by chosen edges to the vertices in $L_{v,i}$ is at most $g_i(v)$.
    
    For a fixed vertex $v$, the family $\{ L_{v,i} \colon i\in[k] \}$ forms a laminar family, since the sets partition the incident edges by color. Hence all degree constraints can be expressed through a laminar family of vertex sets, and the problem on $G'$ becomes an instance of the weighted hierarchical $b$-matching problem, and thus our reduction is complete. By \cite[Corollary~3]{madarasi2021}, this problem is solvable in strongly polynomial time.
    
    Finally, observe that $|V'|=\sum_{v\in V} |\delta_G(v)|$ and $|E'|=|E|$, so the size of $G'$ is linear in $|E|$ and therefore polynomial in the input size. This yields a strongly polynomial algorithm for our original problem. 
\end{proof}

\section{Conclusions}
\label{sec:conc}

In this paper, we introduced and studied the Maximum-size Degree Bounded Matroid Independent Set problem, which in turn led to improved approximation guarantees for the Maximum-size Properly Colored Forest problem. We also extended our framework to other properly colored structures, including properly colored forests with bundles, branchings, and $b$-matchings.

We close the paper by mentioning some open problems:
\begin{enumerate}\itemsep0em
    \item For \maxwdbmis, the ratio $\ln 4/(\Delta+2)$ improves on the bound $1/\Delta$ for large $\Delta$, but is weaker for all $\Delta \le 5$. Establishing a polynomial-time $1/\Delta$-approximation for $\Delta \le 5$ remains open. A broader goal is to break the $1/\Delta$ barrier altogether -- most notably, to obtain an approximation ratio exceeding $1/2$ for \maxwpf.
    \item For strongly base orderable matroids, \maxwdbmis is a $(\Delta+1)$-exchange system and admits a $2/(\Delta+2)$-approximation. Determining additional matroid classes that admit an approximation ratio better than $(\ln 4)/(\Delta+2)$ remains open.
    \item The Properly Colored Arborescence Conjecture is a relaxation of the Rainbow Arborescence Conjecture. Establishing the properly colored version remains open and would mean a significant progress toward the rainbow conjecture.
\end{enumerate}

\medskip
\paragraph{Acknowledgement.}
Yuhang Bai was supported by the National Natural Science Foundation of China (12131013 and 12471334), by China Scholarship Council (202406290002), and by Shaanxi Fundamental Science Research Project for Mathematics and Physics (22JSZ009). This research has been implemented with the support provided by the Lend\"ulet Programme of the Hungarian Academy of Sciences (LP2021-1/2021), by the Ministry of Innovation and Technology of Hungary from the National Research, Development and Innovation Fund (ADVANCED 150556 and ELTE TKP 2021-NKTA-62), and by Dynasnet European Research Council Synergy project (ERC-2018-SYG 810115).

\paragraph{Conflicts of interest} The authors declare that they have no known competing financial interests or personal relationships that could have appeared to influence the work reported in this paper.

\bibliographystyle{abbrv}
\bibliography{maxwpf}
\end{document}